\documentclass[a4paper,UKenglish,cleveref, autoref, thm-restate]{lipics-v2021}
\title{Heuristic computation of exact treewidth}

\author{Hisao Tamaki}
{Meiji University, Department of Computer Science, 
Japan} {hisao.tamaki@gmail.com} {https://orcid.org/0000-0001-7566-8505}{}

\authorrunning{Hisao Tamaki}
\Copyright{Hisao Tamaki} 



\category{} 

\relatedversion{} 

\EventEditors{John Q. Open and Joan R. Access}
\EventNoEds{2}
\EventLongTitle{42nd Conference on Very Important Topics (CVIT 2016)}
\EventShortTitle{CVIT 2016}
\EventAcronym{CVIT}
\EventYear{2016}
\EventDate{December 24--27, 2016}
\EventLocation{Little Whinging, United Kingdom}
\EventLogo{}
\SeriesVolume{42}
\ArticleNo{23}
\renewcommand{\AA}{{\cal A}}
\newcommand{\TT}{{\cal T}}
\newcommand{\BB}{{\cal B}}

\newcommand{\XX}{{\cal X}}
\newcommand{\tw}{{\mathop{\rm tw}}}

\newcommand{\lb}{{\mathop{\rm lb}}}

\newcommand{\ncs}{{\mathop{\rm ncs}}}
\newcommand{\nc}{{\mathop{\rm nc}}}
\mathchardef\mhyphen="2D

\newcommand{\BreakFill}{{\sc BreakFill}}
\newcommand{\Lift}{{\sc Lift}}
\newcommand{\InitialSolution}{{\sc InitialSolution}}
\newcommand{\Improve}{{\sc Improve}}
\newcommand{\Merge}{{\sc Merge}}
\newcommand{\ML}{{\mathop{\rm ML}}}
\newcommand{\COV}{{\mathop{\rm COV}}}
\newcommand{\hide}[1]{}
\nolinenumbers

\begin{document}
\maketitle

\begin{abstract}
We are interested in computing the treewidth $\tw(G)$ of
a given graph $G$. Our approach is to design heuristic
algorithms for computing a sequence of improving upper bounds and 
a sequence of improving lower bounds, which would hopefully 
converge to $\tw(G)$ from both sides.
The upper bound algorithm extends and simplifies Tamaki's
unpublished work on a heuristic use
of the dynamic programming algorithm for deciding treewidth due to 
Bouchitt\'{e} and Todinca.
The lower bound algorithm is based on the well-known fact
that, for every minor $H$ of $G$, we have $\tw(H) \leq \tw(G)$.
Starting from a greedily computed minor $H_0$ of $G$, the algorithm
tries to construct a sequence of minors $H_0$, $H_1$, \ldots $H_k$
with $\tw(H_i) < \tw(H_{i + 1})$ for $0 \leq i < k$ and
hopefully $\tw(H_k) = \tw(G)$. 

We have implemented a treewidth solver based on this approach and
have evaluated it on the bonus
instances from the exact treewidth track of PACE 2017 algorithm 
implementation challenge. The results show 
that our approach is extremely effective in tackling instances that
are hard for conventional solvers.
Our solver has an additional advantage over conventional ones
in that it attaches a compact certificate to the lower bound
it computes.
\end{abstract}
\newpage
\section{Introduction}
\label{sec:intro}
Treewidth is a graph parameter which
plays an essential role in the graph minor theory 
\cite{robertson1986graph,robertson1995graph,robertson2004graph}
and is an indispensable tool in designing graph algorithms 
(see, for example, a survey \cite{bodlaender2006treewidth}).
See Section~\ref{sec:prelim} for the definition of treewidth and
tree-decompositions. Let $\tw(G)$ denote the treewidth of graph $G$.
Deciding if $\tw(G) \leq k$ for given $G$ and $k$ 
is NP-complete \cite{arnborg1987complexity}, 
but admits a fixed-parameter linear time algorithm \cite{bodlaender1996linear}.

Practical algorithms for treewidth have also been actively studied
\cite{bodlaender2010treewidth,bodlaender2011treewidth,bannach2017jdrasil,tamaki2019positive,
tamaki2019computing,althaus2021ontamaki}, with recent progresses 
stimulated by PACE 2016 and 2017 \cite{dell2018pace} 
algorithm implementation challenges. The modern treewidth solvers
use efficient implementations of the dynamic programming algorithm
due to Boudhitt\'{e} and Todinca (BT) \cite{bouchitte2001treewidth}.
After a first leap in that direction \cite{tamaki2019positive},
some improvements have been reported
\cite{tamaki2019computing,althaus2021ontamaki}, 
but those improvements are incremental.
 
In this paper, we pursue a completely different approach.
We develop a heuristic algorithm for the upper bound as well
as one for the lower bound. These algorithms
iteratively improve the bounds in hope that they
converge to the exact treewidth from both sides.

Our upper bound algorithm is based on the following idea.
For $\BB \subseteq 2^{V(G}$, 
we say that $\BB$ admits a tree-decomposition of $G$ if
every bag of this tree-decomposition belongs to $\BB$.
The treewidth of $G$ with respect to $\BB$, denoted by $\tw_{\BB}(G)$,
is the smallest $k$ such that $\BB$ admits a tree-decomposition of
$G$ of width $k$. If $\BB$ admits no tree-decomposition of $G$,
then $\tw_{\BB}(G)$ is undefined. A vertex set $X
\subseteq V(G)$ is a \emph{potential maximal clique} of $G$ if it 
is a maximal clique of some minimal triangulation of $G$.
We denote by $\Pi(G)$ the set of all potential maximal cliques
of $G$.  Bouchitt\'{e} and Todinca \cite{bouchitte2001treewidth}
observe that $\Pi(G)$ admits a tree-decomposition of $G$
of width $\tw(G)$ and present a dynamic programming algorithm
(BT dynamic programming) to compute $\tw(G)$ based on this fact.
Indeed, BT dynamic programming can be applied to an arbitrary set
$\Pi$ of potential maximal cliques to compute $\tw_{\Pi}(G)$.
This allows us to work in a solution space where each solution
is a set of potential maximal cliques rather than an individual
tree-decomposition. A solution $\Pi$ encodes a potentially
exponential number of tree-decompositions it admits and 
offers rich opportunities of improvements in terms
of $\tw_{\Pi}(G)$.
This approach has been proposed by Tamaki in his unpublished work 
\cite{tamaki2019heuristic}, 
where he presents several {\it ad hoc} operations to enrich $\Pi$
in hope of reducing $\tw_{\Pi}(G)$. 
We extend and simplify this approach by replacing those operations
with a single merging operation: given two sets $\Pi_1$ and $\Pi_2$
of potential maximal cliques, we construct a new set $\Pi$
that includes $\Pi_1 \cup \Pi_2$ together with some additional
potential maximal cliques potentially useful for making 
$\tw_\Pi(G)$ smaller than both $\tw_{\Pi_1}(G)$ and 
$\tw_{\Pi_2}(G)$. See Section~\ref{sec:upper} for more details. 

The lower bound algorithm is based on the well-known fact
that, for every minor $H$ of $G$, we have $\tw(H) \leq \tw(G)$.
Starting from a greedily computed minor $H_0$ of $G$, the algorithm
tries to construct a sequence of minors $H_0$, $H_1$, \ldots $H_k$
with $\tw(H_i) < \tw(H_{i + 1})$ for $0 \leq i < k$ and
hopefully $\tw(H_k) = \tw(G)$. Although minors have been 
used to compute lower bounds on the treewidth
\cite{bodlaender2011treewidth}, the goal has been to quickly obtain a
lower bound of reasonable quality to be used, say, in branch-and-bound procedures.
There seems to be no attempt in the literature to develop an
algorithm which, given a minor $H$ of $G$ with $\tw(H) < \tw(G)$,
construct a minor $H'$ of $G$ with 
an improved lower bound $\tw(H') > \tw(H)$.
In view of this lack of attempts, our finding that this task
can be performed with reasonable efficiency in practice 
might be somewhat surprising. See Section~\ref{sec:lower} 
for details. 

We have implemented a treewidth solver based on this approach 
and evaluated it on the bonus instance set from the 
PACE 2017 algorithm
implementation challenge for treewidth. 
This set is designed to remain challenging for
solvers to be developed after the challenge. It consists
of 100 instances and, 
according to the summary provided with the set, the time
spent to compute the exact treewidth by the winning solvers
of PACE 2017 is longer than an hour for 58 instances
and longer than 12 hours for 23 instances, 
including 9 instances which fail to be solved at all.
The results of applying our solver on the
91 solved instances are summarized as follows.
With a timeout of 30 minutes using two threads (one for the upper bound and the other
for the lower bound), 62 instances are exactly solved;
for 20 of the other instances, 
the upper bound equals the exact treewidth
and the lower bound is off by one. Moreover, 
with a timeout of 6 hours, our solver exactly solves 2 of 
the 9 unsolved instances. 
These results suggest that our approach is extremely effective
in coping with instances that are hard for conventional solvers.
See Section~\ref{sec:experiments} for details.

The source code of the solver used in our experiments is available at
\cite{twalgor}.

\section{Preliminaries}
\label{sec:prelim}

\paragraph*{Graph notation}
In this paper, all graphs are simple, that is, without self loops or
parallel edges. Let $G$ be a graph.
We denote by $V(G)$ the vertex set
of $G$ and by $E(G)$ the edge set of $G$.
As $G$ is simple, each edge of $G$ is  
a subset of $V(G)$ with exactly two members that
are adjacent to each other in $G$.
The \emph{complete graph} on $V$, denoted
by $K(V)$, is a graph with vertex set $V$ 
in which every vertex is adjacent to all other vertices.
The subgraph of $G$ induced by $U \subseteq V(G)$ is denoted by 
$G[U]$.  We sometimes use an abbreviation $G \setminus U$ to
stand for $G[V(G) \setminus U]$. 
A vertex set $C \subseteq V(G)$ is a \emph{clique} of $G$ if
$G[C]$ is a complete graph.
For each $v \in V(G)$, $N_G(v)$ denotes the set of neighbors of $v$ in $G$:
$N_G(v) = \{u \in V(G) \mid \{u, v\} \in E(G)\}$.
For $U \subseteq V(G)$, the {\em open neighborhood of $U$ in $G$}, denoted
by $N_G(U)$,  is the set of vertices adjacent to some vertex in $U$ but not
belonging to $U$ itself: $N_G(U) = (\bigcup_{v \in U} N_G(v)) \setminus U$.
The {\em closed neighborhood of $U$ in $G$}, denoted by $N_G[U]$, is defined
by $N_G[U] = U \cup N_G(U)$.  

We say that vertex set $C \subseteq V(G)$ is {\em connected in}
$G$ if, for every $u, v \in C$, there is a path in $G[C]$ between $u$
and $v$. It is a {\em connected component} or simply a {\em component} 
of $G$ if it is connected and is inclusion-wise maximal subject to this condition.
A vertex set $S \subseteq V(G)$ is a {\em separator} of $G$ if 
$G \setminus S$ has more than one component.
A graph is a \emph{cycle} if it is connected and every vertex
is adjacent to exactly two vertices. A graph is a \emph{forest}
if it does not have a cycle as a subgraph. A forest is
a \emph{tree} if it is connected.

\paragraph*{Tree-decompositions}
A {\em tree-decomposition} of $G$ is a pair $(T, \XX)$ where $T$ is a tree
and $\XX$ is a family $\{X_i\}_{i \in V(T)}$ of vertex sets of $G$, indexed
by the nodes of $T$, such that the following 
three conditions are satisfied. We call 
each $X_i$ the {\em bag} at node $i$.   
\begin{enumerate}
  \item $\bigcup_{i \in V(T)} X_i = V(G)$.
  \item For each edge $\{u, v\} \in E(G)$, there is some $i \in V(T)$
  such that $u, v \in X_i$.
  \item For each $v \in V(G)$, the set of nodes $I_v = \{i \in V(T) \mid v \in
  X_i\} \subseteq V(T)$ is connected in $T$.
\end{enumerate}
The {\em width} of this tree-decomposition is $\max_{i \in V(T)} |X_i| - 1$.
The {\em treewidth} of $G$, denoted by $\tw(G)$ is the smallest $k$ such
that there is a tree-decomposition of $G$ of width $k$.

It is well-known that, for each pair $(i, j)$ of adjacent
nodes of a tree-decomposition $\TT = (T, \XX)$, the intersection
$X_i \cap X_j$ is a separator of $G$. We say that $\TT$ \emph{induces}
separator $S$ if there is an adjacent pair $(i, j)$ such that
$S = X_i \cap X_j$.

\paragraph*{Minimal separators and potential maximal cliques}
Let $G$ be a graph and $S$ a separator of $G$. 
For distinct vertices $a, b \in V(G)$, 
a separator $S$ is an {\em $a$-$b$
separator} if there is no path between $a$ and $b$ in $G \setminus S$; 
it is a 
{\em minimal $a$-$b$ separator} if it is an $a$-$b$ separator and
no proper subset of $S$ is an $a$-$b$ separator.
A separator is a {\em minimal separator} if it is a minimal $a$-$b$ separator
for some $a, b \in V(G)$.  

Graph $H$ is {\em chordal} if every induced cycle of $H$ has exactly three vertices. 
$H$ is a {\em triangulation of graph $G$} if it is chordal,
$V(G) = V(H)$, and
$E(G) \subseteq E(H)$. A triangulation $H$ of $G$ is {\em minimal}
if it there is no triangulation $H'$ of $G$ such that
$E(H')$ is a proper subset of $E(H)$.
A vertex set $X \subseteq V(G)$ is a {\em potential maximal
clique} of $G$, if $X$ is a maximal clique in 
some minimal triangulation of $G$.
We denote by $\Pi(G)$ the set of all potential maximal  
cliques of $G$ and by $\Pi_k(G)$ the set of all potential maximal cliques
of $G$ of cardinality at most $k$.

\paragraph*{Bouchitt\'{e}-Todinca dynamic programming}
The treewidth algorithm of 
Bouchitt\'{e} and Todinca \cite{bouchitte2001treewidth} 
is based on the fact that every graph $G$ has a minimal triangulation
$H$ such that $\tw(H) = \tw(G)$ (see \cite{heggernes2006minimal} for a clear
exposition). This fact straightforwardly implies that
$\Pi(G)$ admits an optimal tree-decomposition of $G$.
Their algorithm consists of an algorithm for constructing
$\Pi(G)$ and a dynamic programming algorithm (BT dynamic programming) to
compute $\tw_{\Pi(G)}(G)$. As noted in the introduction, BT dynamic
programming can be applied to compute $\tw_{\Pi}(G)$ for an arbitrary
$\Pi \subseteq \Pi(G)$.

The most time-consuming part of their treewidth algorithm 
is the construction of $\Pi(G)$. 
Empirically observing that $\Pi_{k + 1}(G)$
is substantially smaller than $\Pi(G)$ for $k \leq \tw(G)$, 
authors of modern implementations of 
the BT algorithm \cite{tamaki2019positive,tamaki2019computing,althaus2021ontamaki} 
use BT dynamic programming with $\Pi = \Pi_{k + 1}(G)$ 
to decide if $\tw(G) \leq k$. Moreover, they try to avoid the full generation of 
$\Pi_{k + 1}(G)$, by being lazy and generating a potential maximal clique only
when it becomes absolutely necessary in evaluating the recurrence.

Both our upper and lower bound algorithms use, as a subprocedure, 
such an implementation of the BT algorithm for treewidth, in particular an implementation
of the version proposed in \cite{tamaki2019computing}. 
In addition, our upper bound algorithm uses BT dynamic programming
in its fully general form, to evaluate each solution $\Pi$ in our
solution space as described in the introduction.
The efficiency of BT dynamic programming, which runs 
in time linear in $|\Pi|$ with a factor polynomial in $|V(G)|$,
is crucial in our upper bound algorithm

\paragraph*{Contractions and minors}
Let $F \subseteq E(G)$ be a forest on $V(G)$.
The \emph{contraction} of $G$ by $F$, denoted by $G / F$
is a graph whose vertices are the connected components of $F$ and
two components $C_1$ and $C_2$ are adjacent to each other
if and only if there is $v_1 \in C_1$ and $v_2 \in C_2$ such
that $v_1$ and $v_2$ are adjacent to each other in $G$.
A graph $H$ is a \emph{minor} of $G$ if it is a subgraph of
some contraction of $G$. It is well-known and is easy to verify
that $\tw(H) \leq \tw(G)$ if $H$ is a minor of $G$.  

\paragraph*{Minimal triangulation algorithms}
There are many algorithms for minimal triangulation of a graph
(see \cite{heggernes2006minimal} for a survey). 
For purposes in the current work, we are interested in algorithms that produce
a minimal triangulation of small treewidth.
Although a minimal triangulation $H$ of $G$ such that
$\tw(H) = \tw(G)$ can be computed by the BT algorithm for treewidth,
we need a faster heuristic algorithm. 
The MMD (Minimal Minimum Degree) 
algorithm \cite{berry2003minimum} is known to
perform well.
We use a variant MMAF (Minimal Minimum Average Fill)
\cite{tamaki2021heuristic} of MMD which 
performs slightly better than the original MMD 
on benchmark instances.

\paragraph*{Safe separators and almost-clique separators}
Bodlaender and Koster \cite{bodlaender2006safe} introduced
the notion of safe separators for treewidth.
Let $S$ be a separator of a graph $G$.
We say that $S$ is \emph{safe for width $k$}, 
if $S$ is induced by some tree-decomposition of $G$ of width $k$.
It is simply safe if it is safe for width $\tw(G)$.
The motivation of looking at safe separators is the fact
that there are easily verifiable sufficient conditions for 
$S$ being safe and a safe separator detected by
those sufficient conditions can be used to reduce the
problem of deciding if $\tw(G) \leq k$ to smaller subproblems.
A trivial sufficient condition is that
$S$ is a clique. Bodlaender and Koster observed that
this condition can be relaxed to $S$ being an \emph{almost-clique},
where $S$ is an almost-clique if $S \setminus \{v\}$ is
a clique for some $v \in S$. More precisely, a minimal separator
that is an almost-clique is safe. They showed that
this observation leads to a powerful preprocessing method of treewidth computation.
An \emph{almost-clique separator decomposition} of graph $G$ is a
tree-decomposition $\AA$ of $G$ such that 
every separator induced by $\AA$ is an almost-clique minimal separator.
For each bag $A_i$ of $\AA$, let $G_i$ be a graph
on $A_i$ obtained from $G[A_i]$ by adding edges of $K(N(C))$ 
for every component $C$ of $G \setminus A_i$. The following proposition
holds \cite{bodlaender2006safe}.

\begin{proposition}
\label{prop:acs-dec}
\begin{enumerate}
  \item $\tw(G)$ is the maximum of $\tw(G_i)$, where $i$ ranges
  over the nodes of the decomposition, and a tree-decomposition of
  $G$ of width $\tw(G)$ is obtained by combining tree-decompositions
  of $G_i$ as prescribed by $\AA$.
  \item $G_i$ is a minor of $G$ for each $i$.
\end{enumerate}
\end{proposition} 

Unpublished work of Tamaki \cite{tamaki2021heuristic} 
shows that this preprocessing approach is effective for
instances that are much larger than those tested in \cite{bodlaender2006safe},
using a heuristic method for constructing almost-clique separator
decompositions. We use his implementation in the current work.

We also make an unconventional use of safe separators in our
lower bound algorithm. When we have a lower bound of $k$ on
$\tw(G)$, we wish to evaluate a minor $H$ of $G$ for the possibility
of leading to a stronger lower bound. We use the set of
all minimal separators of $G$ that are safe for width $k$
in this evaluation. Note that the computation of such a set
is possible because $H$ is small.  

\section{The upper bound algorithm}
\label{sec:upper}
Recall that $\Pi(G)$ denotes the set of 
all potential maximal cliques of
$G$. In our upper bound algorithm, a \emph{solution} for $G$ is
a subset $\Pi$ of $\Pi(G)$ that admits at least one tree-decomposition of
$G$ and the \emph{value} of solution $\Pi$ is $\tw_{\Pi}(G)$. 

Our algorithm starts from a greedy solution and iteratively improves
the solution.
To improve a solution $\Pi$, we \emph{merge} it with another solution
$\Omega$ into another solution $\Pi'$ in hope of having
$\tw_{\Pi}(G) < \min\{\tw_{\Pi}(G), \tw_{\Omega}(G)\}$. 
This merged solution
$\Pi'$ contains $\Pi \cup \Omega$ together with some 
other potential maximal cliques so that
$\Pi'$ would admit a tree-decomposition that contains
some bags in $\Pi$, some bags in $\Omega$, and some bags
belonging to this additional set of potential maximal cliques.
We describe below how these additional potential maximal cliques are
computed.

Let $X \in \Pi$ and $Y \in \Omega$ be distinct and not
crossing each other. Then, there is a unique component $C$ of $G \setminus X$
such that $Y \subseteq N[C]$ and a unique component $D$ of $G \setminus Y$ such
that $X \subseteq N[D]$. Let $U = N[C] \cap N[D]$ and let
$H$ be a graph on $U$ obtained from $G[U]$ by adding 
edges of $K(N(B))$ for each component $B$ of $G \setminus U$.
Let $\hat{H}$ be a minimal triangulation of $H$ with small treewidth:
if $|U|$ does not exceed a fixed threshold BASE\_SIZE (=60), then we
use the BT algorithm for treewidth to compute
$\hat{H}$; otherwise we use MMAF to compute $\hat{H}$. 
Here, BASE\_SIZE is a parameter of the algorithm
represented as a constant in the implementation. The parenthesized
number is the value of this parameter used in our experiment. In the
following, we use the same convention for citing algorithm parameters.
Note that each maximal clique of $\hat{H}$ 
is either a potential maximal clique of $G$ or a minimal separator of
$G$. If $\tw(\hat{H}) \leq \tw_{\Pi}(G)$, then we add all
potential maximal cliques of $G$ that are maximal cliques of
$\hat{H}$ to $\Pi'$. When this happens, then $\Pi'$ admits
a tree-decomposition of width at most $\max\{\tw_{\Pi}(G),
\tw_\Omega(G)\}$ consisting of some bags in $\Pi$, some bags in $\Omega$,
and some bags that are maximal cliques of $\hat{H}$. In this way,
$\Pi'$ would admit tree-decompositions that are not admitted by 
the simple union $\Pi \cup \Omega$ and, with some luck, some of the
newly admitted tree-decomposition may have width smaller than $\tw_{\Pi}(G)$.

The procedure \Merge($\Pi$, $\Omega$) merges $\Pi$ with $\Omega$, 
applying the above operation to some pairs
$X \in \Pi$ and $Y \in \Omega$ chosen as follows. 
We first pick $X \in \Pi$
uniformly at random. Then, let $C$ be the largest component of 
$G \setminus X$.
We choose $Y \in \Omega$ such that $Y \subseteq N[C]$ and 
$|Y| \leq \tw_{\Pi}(G)$.
The first condition ensures that the method in the previous
paragraph can be applied to this pair of $X$ and $Y$. 
The second condition is meant to increase the chance of
newly admitted tree-decompositions to have width smaller than $\tw_{\Pi}(G)$.
We sort the candidates of such $Y$ in the
increasing order of $|N[C] \cap N[D]|$, where $D$ is the component
of $Y$ such that $X \subseteq N[D]$, and use 
the first N\_TRY (=50) elements of this sorted list.
We prefer $Y$ such that $U = N[C] \cap N[D]$ is small, 
because that would increase the chance of the minimal triangulation
$\hat{H}$, described in the previous paragraph, to have small treewidth.
All the resulting potential maximal
cliques from these trials are added to $\Pi'$. 

In addition to procedure \Merge, our algorithm uses two 
more procedures \InitialSolution() and \Improve($\Pi$) described below.
The input graph $G$ is fixed in these procedures.

\begin{description}
\item[\InitialSolution()] generates an initial solution $\Pi$.
We use a randomize version of MMAF to generate 
N\_INITIAL\_GREEDY( =10) minimal triangulations of $G$ 
and take $H$ with the smallest treewidth.
The solution $\Pi$ returned by the call 
\InitialSolution() is the set of maximal cliques of $H$.
\item[\Improve($\Pi$)] returns a solution $\Pi'$ with $\Pi \subseteq \Pi'$,
where efforts are made to make $\tw_{\Pi'}(G)$ strictly smaller than
$\tw_{\Pi}(G)$.
We proceed in the following steps.
\begin{enumerate}
  \item Let $\Omega = $ \InitialSolution().
  \item While $\tw_{\Omega}(G)) > \tw_{\Pi}(G)$, replace $\Omega$ by
  \Improve($\Omega$).
  \item Return \Merge($\Pi$, $\Omega$). 
\end{enumerate}
Note the solution $\Omega$ to be merged with $\Pi$ is generated independently
of $\Pi$ and improved so that $\tw_{\Omega}(G) \leq \tw_{\Pi}(G)$ before being
merged with $\Pi$. 
\end{description} 

Given these procedures, the main iteration of our algorithm proceeds as follows.
It is supposed that the algorithm has an access to lower bounds
provided by the lower bound algorithm.
\begin{enumerate}
  \item Let $\Pi = $ \InitialSolution(). Report $\tw_{\Pi}(G)$ as
  the initial upper bound on $\tw(G)$, together with a
  tree-decomposition of $G$ of width $\tw_{\Pi}(G)$ admitted by $\Pi$. 
  \item While $\tw_{\Pi}(G)$ is greater than the current lower bound,
  replace $\Pi$ by \Improve($\Pi$). When this replacement reduces $\tw_{\Pi}(G)$,
  report this new upper bound on $\tw(G)$, together with a 
  tree-decomposition of $G$ of width $\tw_{\Pi}(G)$ admitted by $\Pi$.
  We also shrink $\Pi$, removing all members of cardinality greater than $k + 2$,
  whenever $\tw_{\Pi}(G)$ is improved to $k$.
\end{enumerate}

\section{The lower bound algorithm}
\label{sec:lower}
In our lower bound algorithm, we use a procedure we call 
\Lift, which, given a graph $G$ 
and a forest $F$ on $V(G)$, 
finds another forest $F'$ such that $\tw(G / F') > \tw(G / F)$;
it inevitably fails if $\tw(G / F) = \tw(G)$. 
Given this procedure, the overall lower bound algorithm proceeds in the
following steps.
Let $G$ be given. It is supposed that the algorithm has an access to
the upper bound being computed by the upper bound algorithm.

\begin{enumerate}
  \item Construct a contraction $G / F$ of $G$, using a greedy heuristic
  for contraction-based lower bounds on treewidth.
  \item While $\tw(G / F)$ is smaller than the current upper bound on $\tw(G)$,
  replace $F$ by \Lift($G$, $F$) unless this call fails.  
\end{enumerate}

When the current upper bound is larger than $\tw(G)$, it is possible
that the call \Lift($G$, $F$) is made for $F$ such that $\tw(G / F) = \tw(G)$.
In such an event, the call would eventually fail but the time it takes would be
at least as the time taken by conventional solvers to compute $\tw(G)$.
Our solver implements a mechanism to
let such a call terminate as soon as the upper bound is updated to be equal
to the current lower bound $\tw(G / F)$.

The design of procedure {\Lift} is described in the following subsections.

\subsection{Contraction lattice}
First consider looking for the result of \Lift($G$, $F$) among the
subsets of $F$. Assuming that $\tw(G / F) < \tw(G)$, a subset $F'$ of
$F$ such that $\tw(G / F') > \tw(G / F)$ certainly exists.

For each $A \in 2^F$, let $H_A$ denote the contraction 
$G / (F \setminus A)$.
Then, $\Lambda(G, F) = \{H_A \mid A \in 2^F\}$ is a lattice
isomorphic to the power set lattice $2^F$, with top $G$ and bottom $G / F$.
Brute force searches for $H$ with $\tw(H) > \tw(G / F)$ in this lattice
are hopeless as $|F|$ can be
large: we typically have $|F| > 100$ for graph instances we target.

We need to understand the terrain of this lattice to design an
effective search method. In the remainder of this subsection
and subsequent subsections, let $k$ denote $\tw(G / F)$. 
Call $H \in \Lambda(G, F)$ \emph{lifted} if $\tw(H) > k$; otherwise
call it \emph{unlifted}. Let $\ML(G, F)$ denote the set of 
minimal lifted elements of $\Lambda(G, F)$. 
Call an unlifted element $H_A$ \emph{covered} if there
is some $H_B \in \ML(G, F)$ with $A \subset B$.
Let $\COV(G, F)$ denote the set of all covered unlifted elements of
$\Lambda(G, F)$. Ideally,
we wish to confine our search in $\COV(G, F) \cup \ML(G, F)$.
This would be possible if there is a way of knowing, for
each covered element $H_A$ and $e \in F \setminus A$, if
$H_{A \cup \{e\}}$ is still covered. Then, we would start
with the clearly covered element $H_{\emptyset}$ and 
greedily ascend the lattice staying in $\COV(G, F)$ until
we hit an element in $\ML(G, F)$. Although such an ideal
scenario is unlikely to be possible, we still aim at something close to it
in the following sense. For each $A \in 2^F$ such that
$H_A$ is unlifted, call $S \subseteq A$ an \emph{excess} in
$A$ if $H_{A \setminus S}$ is covered. We wish to confine our search
among elements with a small excess. We employ a strategy that
works only for pairs $(G, F)$ with a special property, which is described
in the following subsections.

\subsection{Critical fills}
Call a pair $\{u, v\}$ of distinct vertices of $G$ a \emph{fill} of $G$ if 
$u$ and $v$ are not adjacent to each other in $G$. For a fill $e$ of $G$,
let $G + e$ denote the graph on $V(G)$ with edge set $E(G) \cup \{e\}$.
We say that a fill $e$ of $G$ is \emph{critical for $F$} if
$\tw((G + e) / F) > \tw(G / F)$.

Suppose $G$ has a critical fill for $F$. Observe the following.

\begin{enumerate}
  \item Since adding a single edge to $G / F$ increases its treewidth, a small
  number of uncontractions applied to $G / F$ may suffice to increase its
  treewidth. Thus, we may expect that there is a lifted
  element $H_A$ in the lattice $\Lambda(G, F)$ such that 
  $|A|$ is smaller than the value that is expected in a general case (without
  a critical fill). This would make our search for a lifted element easier. 
  \item The fact that $e$ is a critical fill could be used to guide our
  search for a lifted element.
\end{enumerate}

The next section describes how we exploit the existence of a critical fill
to guide our search.

\subsection{Breaking a critical fill} 
Assuming that $G$ has a fill $e = \{u, v\}$ critical for $F$, 
we look for a lifted element $H_A$ in the lattice $\Lambda(G, F)$. 
We call the procedure for this operation \BreakFill($F$, $e$).
The reason of this naming is that, informally speaking, we 
remove the fill $e$ from $G + e$ by uncontracting some edges in $F$
maintaining the treewidth. We need some preparations before 
we describe this procedure.

\begin{proposition}
\label{prop:critical}
If $H_A$ is unlifted then $e$ is critical for $F \setminus A$.
\end{proposition}
\begin {proof}
Since $H_A$ is unlifted, we have 
$\tw(G / (F \setminus A)) = k$.
So, it suffices to show that $\tw((G + e) / (F \setminus A)) > k$.
We have $\tw((G + e) / F) > k$, 
since $e$ is critical for $F$.
We also have $\tw((G + e) / (F \setminus A)) \geq 
\tw((G + e) / F)$, since $(G + e) / F$ is a contraction
of $(G + e) / (F \setminus A)$. Therefore we have
$\tw((G + e) / (F \setminus A)) > k$.
\end{proof}

Let $u_A$ ($v_A$) denote
the vertex of $H_A$ into which $u$ ($v$, respectively) is contracted. 
We say that a separator $S$ of $H_A$ \emph{crosses} $e$ if $u_A$ and $v_A$ belong
to two distinct components of $H_A \setminus S$.
Define $\ncs_k(A, e)$ to be the number of
minimal separators of $H_A$ that are safe for width $k$
and moreover cross $e$. Observe the following.
\begin{enumerate}
  \item We have $\ncs_k(A, e) > 0$ if $H_A$ is unlifted. 
  To see this, suppose otherwise that
  $H_A$ is unlifted but $\ncs_k(A) = 0$.
  Then, $H_A$ has a tree-decomposition $\TT$ of width $k$ such that
  none of the minimal separators induced by $\TT$ crosses $e$. Then, 
  $\TT$ is a tree-decomposition of $(G + e) / (F \setminus A)$ as well, 
  contradicting the assumption that $e$ is critical for $F$ and
  hence for $F \setminus A$ by Proposition~\ref{prop:critical}.
  \item We have $\ncs_k(A, e) = 0$ if $H_A$ is lifted. This is simply
  because $H_A$ does not have any minimal separator that is safe
  for treewidth $k$ if $\tw(H_A) > k$. 
\end{enumerate} 
We empirically observe a tendency that 
$\ncs_k(A, e)$ decreases as $H_A$ approaches
a lifted element from below. Based on this observation, 
we use this function $\ncs_k$ to guide our search for lifted elements.
We are ready to describe our procedure \BreakFill($F$, $e$).
It involves two parameters UNC\_CHUNK ( =5) and 
N\_TRY ( =100) and proceeds as follows. 

\begin{enumerate}
  \item Let $A = \emptyset$.
  \item While $H_A$ is unlifted, repeat the following:
  \begin{enumerate}
    \item Pick N\_TRY random supersets $A'$ of $A$ with cardinality 
    $|A|$ + UNC\_CHUNK (or $|F|$ if this exceeds $|F|$)
    and let $A_{\rm best}$ be $A'$ such that $\ncs_k(A', e)$ is the smallest.
    \item Replace $A$ by $A_{\rm best}$.
  \end{enumerate}
  \item Return $A$.
\end{enumerate}
This procedure is correct in a purely theoretical sense: since $H_F = G$ is
always lifted, provided $\tw(G / F) < \tw(G)$, it eventually returns some $A$ such
that $H_A$ is lifted. The time required for this to happen, however, can be 
prohibitively long since we are supposing that $|F|$ is fairly large.
The success of this procedure hinges on the effectiveness 
of our heuristic relying on the critical fill. 

\subsection{Procedure \Lift}
Our procedure \Lift($G$, $F$) is recursive and works in the following steps.
\begin{enumerate}
  \item Choose a fill $e$ of $G$ with the following heuristic criterion.
  For each $v \in V(G)$, let $d_F(v)$ denote the degree of $v'$ in $G / F$,
  where $v'$ is the vertex of $G / F$ into which $v$ contracts.
  Then we choose $e = \{u, v\}$ so as to maximize the pair $(d_F(u), d_F(v))$
  in the lexicographic ordering, where the order of $u$ and $v$ is chosen
  so that $d_F(u) \leq d_F(v)$.
  \item Let $F_1 = $\Lift($(G + e)/ F$). If $\tw(G / F_1) > \tw(G / F)$ then
  return $F_1$; otherwise, observing that $e$ is critical for $F_1$, 
  call \BreakFill($e$, $F_1$) to find 
  $F_2 \subseteq F_1$ such that $\tw(G / F_2) > \tw(G / F_1)$.
  \item  Greedily compute a maximal forest $F_3$ on $V(G)$ such that $F_2
  \subseteq F_3$ and $\tw(G / F_3) = \tw(G / F_2)$. Return $F_3$.
\end{enumerate}
The criterion for choosing $e$ in the first step is based on 
the following heuristic reasoning. Let $u''$ ($v''$) be the vertex
of $G / F_1$ into which $u$ ($v$, respectively) contracts. We expect that
if the size of the minimum vertex cut 
between $u''$ and $v''$ in $G / F_1$ is large, 
then the minimum cardinality of $A$ such that $H_A$ is lifted
in $\Lambda(G, F_1)$
would have a tendency to be small. 
Indeed, if the size of this cut is as large as $k$, 
then no separator of cardinality at most $k$ crosses $e$ and therefore
we have $\tw(H_\emptyset) > 0$. Although we cannot predict the size
of the minimum $u''$-$v''$ cut in $G / F_1$ when we are choosing $e$, 
having larger degrees of $u'$ and $v'$ in $G / F$ could have some
positive influence toward this goal. 
This criterion, however, has not been evaluated with respect to 
this goal: further experimental studies are needed here.  

We emphasize that Step 3 above is crucial in allowing us to work on
relatively small contractions throughout the entire computation.
Note also that, due to this step, the result of \Lift($G$, $F$) is not
a subset of $F$ in general.

\section{The overall algorithm}
\label{sec:overall}
The upper bound algorithm and the lower bound algorithm, together
with the preprocessing algotirhm, are combined in the following manner.
Fix the graph $G$ given.

\begin{enumerate}
  \item We compute an almost-clique separator decomposition $\AA$
  of $G$ using the method in \cite{tamaki2021heuristic}. 
  \item For each bag $A_i$ of $\AA$, let $G_i$ denote the
  graph on $A_i$ obtained from $G[A_i]$ by adding edges of $K(N(C))$ 
  for each component $C$ of $G \setminus A_i$.
  By Proposition~\ref{prop:acs-dec}, 
  the task of computing $\tw(G)$ reduces to the tasks of
  computing $\tw(G_i)$ for $i$, for all nodes $i$ of $\AA$.
  Moreover, $G_i$ is a minor of $G$.
  \item Let $i^*$ be such that $|A_{i^*}| \geq |A_i|$ for every node 
  $i$ of $\AA$. The lower bound algorithm works on $G_{i^*}$.
  When it finds a new lower bound $\tw(G_{i^*}/F)$ on $\tw(G_{i^*})$,
  this is also a lower bound on $\tw(G)$ since $G_{i^*}$ is a minor of $G$;
  we record this new lower bound $\tw(G_{i^*} /F)$ together with
  the minor $\tw(G_{i^*} /F)$ of $G$ certifying it. 
  \item The upper bound algorithm works on $G_i$ for each $i$, 
  keeping the current solution $\Pi_i$ of $G_i$ for each $i$.
  After initializing $\Pi_i$ for each $i$, we start iterations.
  In each iteration, we chooses $i_0$ to be $i$ such that
  the current upper bound $\tw_{\Pi_i}(G_i)$ on $\tw(G_{i})$ is the largest
  and replace $\Pi_{i_0}$ by \Improve($\Pi_{i_0}$), where
  the implicit graph it works on is set to $G_i$. 
  When the maximum of the upper bounds $\tw_{\Pi_i}(G_i)$ decreases, 
  we record the new upper bound on $G$ together with the
  tree-decomposition of $G$ combining $\AA$ with the 
  currently best tree-decomposition
  of $G_i$ for all nodes $i$ of $\AA$.
  \item As described in the previous sections, the upper bound algorithm
  and the lower bound algorithm have access to the current bound computed
  by each other and terminate when they match. 
\end{enumerate}

\section{Experiments}
\label{sec:experiments}
We have evaluated our solver by experiments.
The computing environment for our experiments is as follows.
CPU: Intel Core i7-8700K, 3.70GHz; RAM: 64GB; 
Operating system: Windows 10Pro, 64bit; 
Programming language: Java 1.8; JVM: jre1.8.0\_271.
The maximum heap size is set to 60GB. The solver uses
two threads, one for the upper bound and the other for the lower bound,
although more threads may be invoked for garbage collection by
JVM. 

As described in the previous sections, both of the upper and lower bound
algorithms use BT dynamic programming for deciding the treewidth, and enumerating the 
safe separators, of small graphs.
Our solver uses an implementation of the semi-PID version of this algorithm
\cite{tamaki2019computing}, which is available at the same github repository
\cite{twalgor} in which the entire source code of our solver is posted.

The upper bound computation uses a single sequence of pseudo-random numbers 
and the lower bound computation uses another independent single sequence.
The initial seed is set to a fixed value of 1 for both of these sequences,
for the sake of reproducibility. With this setting, our solver can be
considered deterministic.
    
We use the bonus instance set from the exact treewidth track of
PACE 2017 algorithm implementation challenge \cite{dell2018pace}. 
This set of instances are available at \cite{pace2017bonus}. We quote the
note by Holgar Dell, the PACE 2017 organizer, explaining his intention 
to compile these instances.  

\begin{quote}
The instance set used in the exact treewidth challenge of PACE 2017 
is now considered to be too easy. 
Therefore, this bonus instance set has been created to offer 
a fresh and difficult challenge. 
In particular, solving these instances in five minutes would require 
a 1000x speed improvement over the best exact treewidth solvers of 
PACE 2017. 
\end{quote}

The set consists of 100 instances and their
summary, available at \cite{pace2017bonus}
in csv format, lists each instance with the time spent for solving it and 
its exact treewidth if the computation is successful. 
According to this summary,  
the exact treewidth is known for 91 of those instances; the remaining
9 instances are unsolved.

Table~\ref{table:instances1} shows the list of those 91 solved instances,
We number them in the increasing order of the computation time provided
in the summary. In the ``name" column, a long instance name is 
shown by a prefix and a suffix making the instance
identifiable. Columns ``n'',
``m'', and ``tw'' give the number of vertices, the number of edges, 
and the treewidth, respectively, of each instance. 
The suffixes ``s'', ``m'', and ``h'' in the time column stand for
seconds, minutes, and hours respectively.

\begin{table}[hbtp]
\caption{bonus instances with known treewidth}
\label{table:instances1}
     \centering
  \begin{tabular}{|r|r|r|r|r|r||r|r|r|r|r|r|}
    \hline
    no. & name & n & m & tw & time & no. & name & n & m & tw & time \\
    \hline
1 & Sz512....man\_3 & 175 & 593 & 14 & 4.74s & 
2 & MD5-3....man\_4 & 225 & 705 & 12 & 7.45s\\
    \hline
3 & Promedas\_69\_9 & 133 & 251 & 9 & 11.5s & 
4 & GTFS\_....ro\_12 & 103 & 212 & 11 & 18.5s\\
    \hline
5 & Promedas\_56\_8 & 155 & 299 & 10 & 28.3s & 
6 & Promedas\_48\_5 & 134 & 278 & 11 & 46.1s\\
    \hline
7 & minxo....man\_2 & 231 & 606 & 4 & 1.24m & 
8 & Promedas\_49\_8 & 184 & 367 & 10 & 1.92m\\
    \hline
9 & FLA\_14 & 266 & 423 & 8 & 3.06m & 
10 & post-....an\_10 & 263 & 505 & 11 & 4.32m\\
    \hline
11 & Pedigree\_11\_7 & 202 & 501 & 14 & 4.64m & 
12 & count....an\_10 & 331 & 843 & 13 & 4.96m\\
    \hline
13 & mrpp\_....man\_3 & 106 & 589 & 24 & 5.00m & 
14 & GTFS\_....ram\_9 & 143 & 303 & 13 & 5.24m\\
    \hline
15 & Promedus\_38\_15 & 208 & 398 & 10 & 5.35m & 
16 & Promedas\_50\_7 & 175 & 362 & 12 & 7.40m\\
    \hline
17 & Promedus\_34\_11 & 157 & 289 & 11 & 7.61m & 
18 & GTFS\_....ro\_15 & 124 & 250 & 13 & 7.73m\\
    \hline
19 & GTFS\_....ro\_14 & 123 & 248 & 13 & 7.88m & 
20 & Promedas\_43\_13 & 197 & 354 & 10 & 8.09m\\
    \hline
21 & Promedas\_46\_8 & 175 & 318 & 11 & 8.36m & 
22 & Promedus\_14\_9 & 173 & 357 & 12 & 9.27m\\
    \hline
23 & jgira....man\_2 & 95 & 568 & 34 & 9.76m & 
24 & modge....man\_2 & 112 & 686 & 35 & 10.0m\\
    \hline
25 & Promedus\_38\_14 & 242 & 462 & 10 & 11.1m & 
26 & Pedigree\_11\_6 & 205 & 503 & 14 & 14.9m\\
    \hline
27 & Promedas\_27\_8 & 165 & 323 & 12 & 15.4m & 
28 & Promedas\_45\_7 & 159 & 313 & 12 & 17.4m\\
    \hline
29 & jgira....man\_2 & 105 & 658 & 33 & 20.4m & 
30 & jgira....man\_2 & 111 & 675 & 33 & 22.2m\\
    \hline
31 & Promedas\_25\_8 & 204 & 378 & 11 & 22.5m & 
32 & Pedigree\_12\_8 & 217 & 531 & 14 & 22.8m\\
    \hline
33 & Promedus\_34\_12 & 210 & 389 & 11 & 26.1m & 
34 & Promedas\_22\_6 & 200 & 415 & 12 & 28.0m\\
    \hline
35 & aes\_2....man\_3 & 104 & 380 & 23 & 31.5m & 
36 & Promedus\_18\_8 & 195 & 411 & 13 & 35.2m\\
    \hline
37 & 6s151.gaifman\_3 & 253 & 634 & 14 & 36.8m & 
38 & LKS\_15 & 220 & 385 & 10 & 39.4m\\
    \hline
39 & Promedas\_23\_6 & 253 & 500 & 12 & 45.3m & 
40 & Promedus\_28\_14 & 193 & 351 & 11 & 47.8m\\
    \hline
41 & Promedas\_21\_9 & 253 & 486 & 11 & 50.9m & 
42 & Promedas\_59\_10 & 209 & 396 & 11 & 56.3m\\
    \hline
43 & Promedas\_60\_11 & 216 & 387 & 11 & 58.2m & 
44 & Promedas\_69\_10 & 194 & 379 & 12 & 1.08h\\
    \hline
45 & newto....man\_2 & 119 & 459 & 19 & 1.18h & 
46 & jgira....man\_2 & 92 & 552 & 34 & 1.22h\\
    \hline
47 & Promedus\_34\_14 & 188 & 352 & 12 & 1.30h & 
48 & modge....man\_2 & 135 & 855 & 33 & 1.33h\\
    \hline
49 & jgira....man\_2 & 100 & 593 & 36 & 1.33h & 
50 & Promedas\_61\_8 & 156 & 305 & 13 & 1.48h\\
    \hline
51 & Promedas\_30\_7 & 164 & 320 & 13 & 1.62h & 
52 & FLA\_13 & 280 & 456 & 9 & 1.66h\\
    \hline
53 & am\_7\_....man\_6 & 189 & 424 & 14 & 1.72h & 
54 & LKS\_13 & 293 & 484 & 9 & 1.78h\\
    \hline
55 & SAT\_d....man\_3 & 130 & 698 & 22 & 1.83h & 
56 & Promedas\_28\_10 & 333 & 605 & 11 & 1.84h\\
    \hline
57 & smtli....man\_6 & 316 & 669 & 13 & 1.89h & 
58 & Promedas\_11\_7 & 191 & 385 & 13 & 2.64h\\
    \hline
59 & Promedus\_20\_13 & 193 & 353 & 12 & 3.39h & 
60 & Pedigree\_13\_12 & 264 & 646 & 15 & 3.50h\\
    \hline
61 & GTFS\_....am\_10 & 187 & 385 & 14 & 3.64h & 
62 & Promedas\_22\_8 & 224 & 441 & 13 & 3.77h\\
    \hline
63 & FLA\_15 & 325 & 522 & 9 & 4.68h & 
64 & GTFS\_....am\_15 & 198 & 406 & 14 & 4.99h\\
    \hline
65 & GTFS\_....am\_13 & 190 & 390 & 14 & 5.35h & 
66 & GTFS\_....am\_12 & 197 & 404 & 14 & 5.46h\\
    \hline
67 & Promedas\_63\_8 & 181 & 374 & 14 & 5.53h & 
68 & GTFS\_....am\_11 & 192 & 395 & 14 & 5.64h\\
    \hline
69 & Pedigree\_12\_14 & 284 & 703 & 15 & 6.18h & 
70 & Promedus\_12\_15 & 293 & 533 & 11 & 6.19h\\
    \hline
71 & Promedus\_12\_14 & 272 & 494 & 11 & 6.63h & 
72 & Promedas\_44\_9 & 276 & 534 & 12 & 6.65h\\
    \hline
73 & Promedas\_32\_8 & 238 & 487 & 13 & 7.05h & 
74 & NY\_13 & 283 & 448 & 9 & 7.16h\\
    \hline
75 & Promedus\_18\_10 & 187 & 397 & 14 & 8.98h & 
76 & Promedas\_34\_8 & 174 & 348 & 14 & 9.11h\\
    \hline
77 & Promedas\_62\_9 & 217 & 427 & 13 & 9.67h & 
78 & Promedus\_17\_13 & 180 & 349 & 13 & 11.1h\\
    \hline
79 & Promedus\_11\_15 & 247 & 497 & 13 & 12.5h & 
80 & Promedas\_24\_11 & 273 & 494 & 12 & 13.4h\\
    \hline
81 & Promedus\_14\_8 & 199 & 417 & 14 & 15.1h & 
82 & am\_9\_9.gaifman\_6 & 212 & 480 & 15 & 15.7h\\
    \hline
83 & NY\_11 & 226 & 369 & 10 & 17.9h & 
84 & mrpp\_....man\_3 & 140 & 856 & 28 & 18.5h\\
    \hline
85 & Pedigree\_12\_10 & 286 & 712 & 16 & 22.1h & 
86 & Promedas\_55\_9 & 221 & 425 & 13 & 27.0h\\
    \hline
87 & Pedigree\_12\_12 & 277 & 683 & 16 & 27.9h & 
88 & Promedas\_46\_15 & 227 & 416 & 13 & 28.4h\\
    \hline
89 & Pedigree\_13\_9 & 268 & 665 & 16 & 38.9h & 
90 & Promedas\_51\_12 & 230 & 431 & 13 & 40.2h\\
    \hline
91 & Promedus\_27\_15 & 189 & 353 & 13 & 41.6h 
\\
    \cline{1-6}
   \end{tabular}
\end{table}
We see that some of these instances are indeed challenging.
Even though they have been solved, they require more than a day
to solve. We also note that, to date, no new solvers
has been published that have overcome the challenge posed by this
instance set.

We have run our solver on these instances with the timeout of
30 minutes. 
Figure~\ref{fig:feasibles1} (instance No.1 -- No.45) and
Figure~\ref{fig:feasibles2} (instance No.46 -- No.91) show the results.
In each column representing an instance $G$, 
the box with a blue number plots the time for 
obtaining the best upper bound computed before timeout, 
where the non-negative number $d$ in the box indicates that 
this bound is $\tw(G) + d$. Similarly, 
the box with a red non-positive number $d$ plots the time for obtaining the best lower 
bound, which is $\tw(G) + d$.

We see from these figures that the bounds computed by our solver are
quite tight for most instances.
Let us say that the result for instance $G$ is of type $(d_1, d_2)$
if the lower bound (upper bound) obtained by our solver is $\tw(G) + d_1$
($\tw(G) + d_2$, respectively). Then, the results are of
type $(0, 0)$ for 62 instances, 
$(-1, 0)$ for 20 instances, $(-2, 0)$ for 3 instances, 
$(-3, 0)$ for one instance, $(0, 1)$ for 4 instances, 
and $(-2, 1)$ for one instance.

We also see from these figures that the performance of our 
solver on an instance is not so strongly correlated to 
the hardness of the instance as measured by the time taken by conventional
solvers. For example, of the last 6 instances,  each of which took more  
than a day to solve by the PACE 2017 solvers, 5 are exactly
solved by our solver and the remaining one has a result of type $(-1, 0)$.
Most of the instances with poorer results, with the
gap of 2 or 3 between the upper and lower bounds, occur
much earlier in the list.

\begin{figure}[htbp]
\begin{center}
\includegraphics[width=6in, bb = 0 0 850 426]{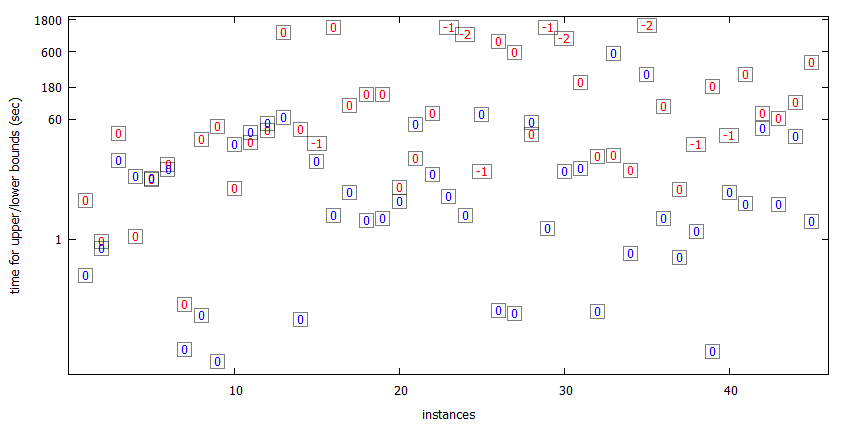}\\
\end{center}
\caption{Time for computing upper/lower bounds for instances 1--45}
\label{fig:feasibles1}
\end{figure} 

\begin{figure}[htbp]
\begin{center}
\includegraphics[width=6in, bb = 0 0 850 426]{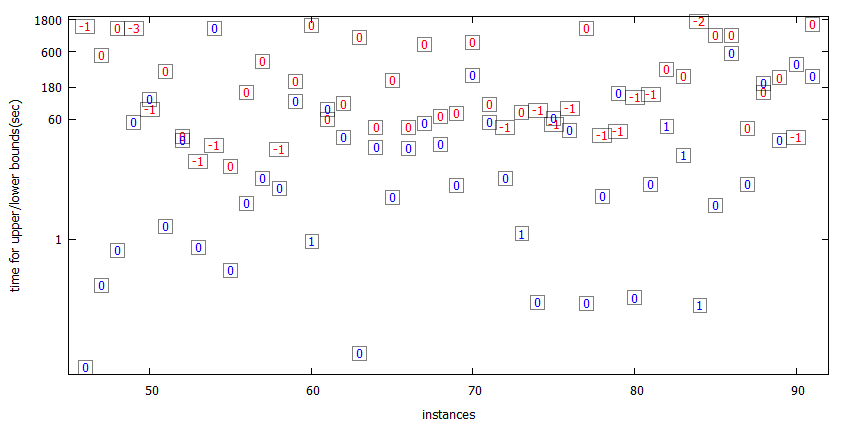}\\
\end{center}
\caption{Time for computing upper/lower bounds for instances 46--91}
\label{fig:feasibles2}
\end{figure} 

We have also run our solver on the 9 unsolved instances of the
bonus set, with 6 hour timeout. It solved 2 of them
and, for other 7 instances, the gap between the upper/lower bounds
is 2 for 2 instances, 3 for 4 instances, and 7 for one instance. 

The certificates of the lower bounds computed by our algorithm
are small and easily verified. For each $G$, let $\lb(G)$ denote the best
lower bound computed by our algorithm before the timeout of 30 minutes and 
let $\nc(G)$ denote the number of vertices of the 
certificate for this lower bound.
Then, the average, the minimum, and the maximum of the ratio $\nc(G)$ / $\lb(G)$
over the 91 solved instances are 2.32, 1.25, and 3.9 respectively. 
The maximum of $\nc(G)$ over these instances is 49 and
the time for verifying each certificate is at most 400 milliseconds.
Figures~\ref{fig:certificates1} and ~\ref{fig:certificates2} 
in the appendix 
show $\lb(G)$, $\nc(G)$, and the time to verify the certificate 
using our implementation of the BT algorithm, 
for each instance $G$.

\section{Conclusions and future work}
\label{sec:conclusions}
Our experiments using the bonus instance set from PACE 2017 algorithm
implementation challenge have revealed that 
our approach to treewidth computation
is extremely effective in tackling instances that are
hard for conventional treewidth solvers.
Even when it fails to give the exact treewidth, it produces
a lower bound very close to the upper bound. In many applications,
such a pair of tight bounds would be satisfactory, since it shows that
further search for a better tree-decomposition could only result in a
small improvement if at all. 

To examine the strength and the weakness
of our approach more closely, evaluation on more diverse sets of 
instances is necessary. For the upper bound part,
there are several implementations of heuristic 
algorithms publicly available, 
such as the submissions to the heuristic treewidth track of
PACE 2017 \cite{pace2017treewidth}.
Although they are primarily intended for large instances
for which exact treewidth appears practically impossible to compute,
some of them are nonetheless potential alternatives 
to our upper bound algorithm. Comparative studies 
would be needed to determine which algorithm is most suitable 
for our purposes. On the other hand, it would also be interesting to
evaluate our upper bound algorithm on large instances that are the 
principal targets of those algorithms. 

Since our lower approach is new, there are several potential improvements
that have not been tried out yet. More work could
result in better performances. 

We may also ask several theoretical questions regarding 
our lower bound approach.
For example, it would be interesting to ask if the lower bound algorithm
can be turned into a fixed parameter tractable algorithm for treewidth.
It would also be interesting and useful to identify parameters or structures 
of graph instances that make them difficult for our
lower bound algorithm.

\subsection*{Acknowledgment}
I thank Holger Dell for posing the
challenging bonus instances, which
have kept defying my ``great ideas'',
showing how they fail, and pointing to
yet greater ideas. 

\bibliography{main}

\begin{thebibliography}{10}

\bibitem{althaus2021ontamaki}
Ernst Althaus, Daniela Schnurbusch, Julian W{\"u}schner, and Sarah Ziegler.
\newblock On tamaki's algorithm to compute treewidths.
\newblock In {\em 19th International Symposium on Experimental Algorithms (SEA
  2021)}. Schloss Dagstuhl-Leibniz-Zentrum f{\"u}r Informatik, 2021.

\bibitem{arnborg1987complexity}
Stefan Arnborg, Derek~G Corneil, and Andrzej Proskurowski.
\newblock Complexity of finding embeddings in a $k$-tree.
\newblock {\em SIAM Journal on Algebraic Discrete Methods}, 8(2):277--284,
  1987.

\bibitem{bannach2017jdrasil}
Max Bannach, Sebastian Berndt, and Thorsten Ehlers.
\newblock Jdrasil: A modular library for computing tree decompositions.
\newblock In {\em 16th International Symposium on Experimental Algorithms (SEA
  2017)}. Schloss Dagstuhl-Leibniz-Zentrum fuer Informatik, 2017.

\bibitem{berry2003minimum}
Anne Berry, Pinar Heggernes, and Genevieve Simonet.
\newblock The minimum degree heuristic and the minimal triangulation process.
\newblock In {\em International Workshop on Graph-Theoretic Concepts in
  Computer Science}, pages 58--70. Springer, 2003.

\bibitem{bodlaender1996linear}
Hans~L Bodlaender.
\newblock A linear-time algorithm for finding tree-decompositions of small
  treewidth.
\newblock {\em SIAM Journal on computing}, 25(6):1305--1317, 1996.

\bibitem{bodlaender2006treewidth}
Hans~L Bodlaender.
\newblock Treewidth: characterizations, applications, and computations.
\newblock In {\em International Workshop on Graph-Theoretic Concepts in
  Computer Science}, pages 1--14. Springer, 2006.

\bibitem{bodlaender2006safe}
Hans~L Bodlaender and Arie~MCA Koster.
\newblock Safe separators for treewidth.
\newblock {\em Discrete Mathematics}, 306(3):337--350, 2006.

\bibitem{bodlaender2010treewidth}
Hans~L Bodlaender and Arie~MCA Koster.
\newblock Treewidth computations i. upper bounds.
\newblock {\em Information and Computation}, 208(3):259--275, 2010.

\bibitem{bodlaender2011treewidth}
Hans~L Bodlaender and Arie~MCA Koster.
\newblock Treewidth computations ii. lower bounds.
\newblock {\em Information and Computation}, 209(7):1103--1119, 2011.

\bibitem{bouchitte2001treewidth}
Vincent Bouchitt{\'e} and Ioan Todinca.
\newblock Treewidth and minimum fill-in: Grouping the minimal separators.
\newblock {\em SIAM Journal on Computing}, 31(1):212--232, 2001.

\bibitem{pace2017treewidth}
Holgar Dell.
\newblock {PACE-challenge/Treewidth}.
\newblock \url{https://github.com/PACE-challenge/Treewidth}, 2017.
\newblock [github repository, accessed January 12, 2022].

\bibitem{pace2017bonus}
Holgar Dell.
\newblock {Treewidth-PACE-2017-bonus-instances}.
\newblock
  \url{https://github.com/PACE-challenge/Treewidth-PACE-2017-bonus-instances/},
  2017.
\newblock [github repository, accessed January 12, 2022].

\bibitem{dell2018pace}
Holger Dell, Christian Komusiewicz, Nimrod Talmon, and Mathias Weller.
\newblock The pace 2017 parameterized algorithms and computational experiments
  challenge: The second iteration.
\newblock In {\em 12th International Symposium on Parameterized and Exact
  Computation (IPEC 2017)}. Schloss Dagstuhl-Leibniz-Zentrum fuer Informatik,
  2018.

\bibitem{heggernes2006minimal}
Pinar Heggernes.
\newblock Minimal triangulations of graphs: A survey.
\newblock {\em Discrete Mathematics}, 306(3):297--317, 2006.

\bibitem{robertson1986graph}
Neil Robertson and Paul~D. Seymour.
\newblock Graph minors. ii. algorithmic aspects of tree-width.
\newblock {\em Journal of algorithms}, 7(3):309--322, 1986.

\bibitem{robertson1995graph}
Neil Robertson and Paul~D Seymour.
\newblock Graph minors. xiii. the disjoint paths problem.
\newblock {\em Journal of combinatorial theory, Series B}, 63(1):65--110, 1995.

\bibitem{robertson2004graph}
Neil Robertson and Paul~D Seymour.
\newblock Graph minors. xx. wagner's conjecture.
\newblock {\em Journal of Combinatorial Theory, Series B}, 92(2):325--357,
  2004.

\bibitem{tamaki2019computing}
Hisao Tamaki.
\newblock Computing treewidth via exact and heuristic lists of minimal
  separators.
\newblock In {\em International Symposium on Experimental Algorithms}, pages
  219--236. Springer, 2019.

\bibitem{tamaki2019heuristic}
Hisao Tamaki.
\newblock A heuristic use of dynamic programming to upperbound treewidth.
\newblock {\em arXiv preprint arXiv:1909.07647}, 2019.

\bibitem{tamaki2019positive}
Hisao Tamaki.
\newblock Positive-instance driven dynamic programming for treewidth.
\newblock {\em Journal of Combinatorial Optimization}, 37(4):1283--1311, 2019.

\bibitem{tamaki2021heuristic}
Hisao Tamaki.
\newblock A heuristic for listing almost-clique minimal separators of a graph.
\newblock {\em arXiv preprint arXiv:2108.07551}, 2021.

\bibitem{twalgor}
Hisao Tamaki.
\newblock {twalgor/tw}.
\newblock \url{https://github.com/twalgor/}, 2022.
\newblock [github repository].

\end{thebibliography}

\newpage
\section*{Appendix}
\begin{figure}[htbp]
\begin{center}
\includegraphics[width=5.5in, bb = 0 0 850 426]{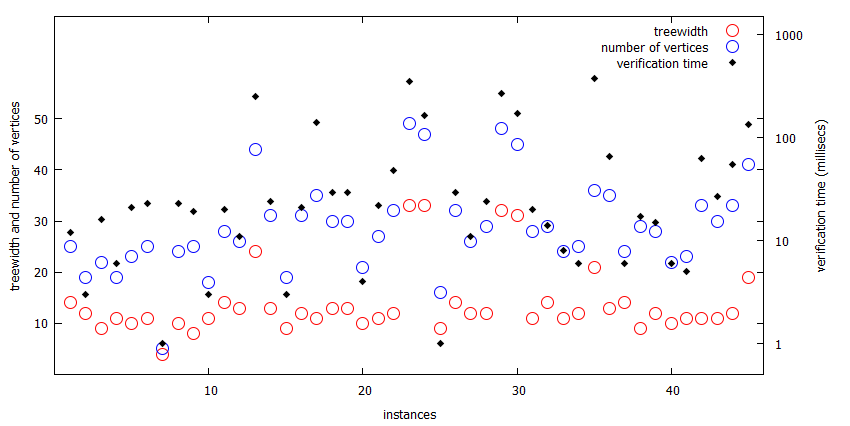}\\
\end{center}
\caption{Lower bound certificates for instances 1--45}
\label{fig:certificates1}
\end{figure} 

\begin{figure}[htbp]
\begin{center}
\includegraphics[width=5.5in, bb = 0 0 850 426]{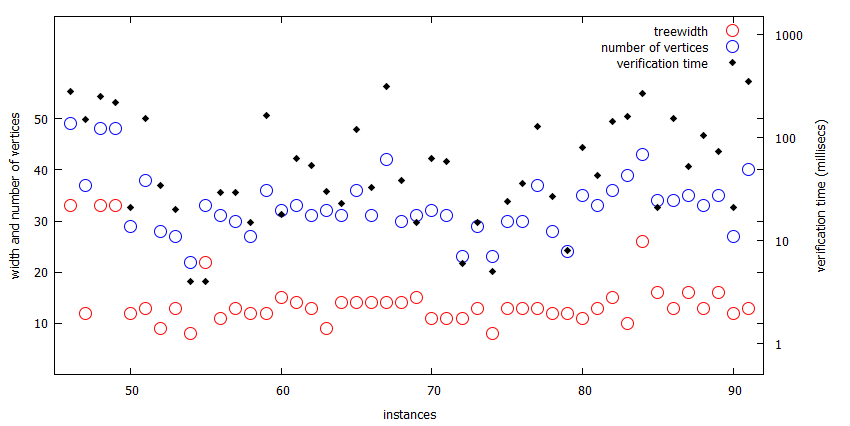}\\
\end{center}
\caption{Lower bound certificates for instances 46--91}
\label{fig:certificates2}
\end{figure}

\end{document}